\documentclass[sigconf]{acmart}

\usepackage{booktabs} 

\usepackage{graphicx}
\usepackage{verbatim}
\usepackage{graphicx}
\usepackage{amsmath}
\usepackage{amssymb}
\usepackage{amsthm}
\usepackage{hyperref}
\allowdisplaybreaks
\usepackage{float}
\setcopyright{rightsretained}

\newcommand{\NN}{\mathrm{NN}}
\newcommand{\dH}{\mathrm{d_H}}
\newcommand{\JS}{\mathrm{JS}}
\newcommand{\tb}{\mathrm{\psi}}
\newcommand{\Sq}{\mathcal{S}_q}
\newcommand{\N}{\mathrm{N}}

\newcommand{\BCS}{\mathrm{BCS}}
\newtheorem{lem}[theorem]{Lemma}
\newtheorem{prop}[theorem]{Proposition}
\newtheorem{remark}[theorem]{Remark}
\newtheorem{note}[theorem]{Note}

\newcommand\numberthis{\addtocounter{equation}{1}\tag{\theequation}}

\acmISBN{123-4567-24-567/08/06}

\acmPrice{15.00}

\begin{document}
\title{Efficient Compression Technique for Sparse Sets}

\author{Rameshwar Pratap}
\affiliation{%
  \institution{}
}
\email{rameshwar.pratap@gmail.com}

\author{Ishan Sohony}
\affiliation{%
  \institution{PICT, Pune}
}
\email{ishangalbatorix@gmail.com}

\author{Raghav Kulkarni}
\affiliation{%
  \institution{Chennai Mathematical Institute}
 }
\email{kulraghav@gmail.com}

\begin{abstract}

Recent technological advancements have led to the generation of huge amounts of data over the web, such as text, image, audio and video. Needless to say, most of this data is high dimensional and sparse, consider, for instance, the bag-of-words representation used for representing text. Often, an efficient search for similar data points needs to be performed  in many applications like clustering, nearest neighbour search, ranking and indexing. Even though there have been significant increases in computational power, a simple brute-force similarity-search on such datasets is inefficient and at times impossible. Thus, it is desirable to get a compressed representation which preserves the similarity between data points.  In this work, we consider the data points as sets and use Jaccard similarity as the similarity measure. Compression techniques are generally evaluated on the   following parameters --1) Randomness required for compression, 2) Time required for compression, 3) Dimension of the data after compression,  and 4) Space required to store the compressed data. 
Ideally, the compressed representation of the data should be such, that the similarity between each pair of data points is preserved, 
while keeping the time and the randomness required for compression as low as possible. 

Recently, Pratap and Kulkarni~\cite{KulkarniP16}, suggested a compression technique for compressing high dimensional, sparse, binary data while preserving the Inner product and Hamming distance between each pair of data points.  In this work, we show that their compression technique also works well for  Jaccard similarity. 
We present a theoretical proof of  the same  and complement  it with rigorous experimentations on synthetic as well as real-world datasets.
We also compare our results with the state-of-the-art "min-wise independent  permutation",  and show that   
our compression algorithm achieves almost equal accuracy while significantly reducing the compression time and the randomness. 
Moreover, after compression our compressed representation is in binary form as opposed to integer in case of min-wise permutation, which leads to a significant reduction in search-time on the compressed data.

\end{abstract}

%
%


\maketitle

\section{Introduction}
 We are at the dawn of a new age. An age in which the availability of raw computational power and massive data sets gives machines the ability to learn, leading to the first practical applications of Artificial Intelligence. 
 The human race has generated more amount of data in the last $2$ years than in the last couple of decades, and it seems like  just the beginning. 
As we can see, practically everything we use on a daily basis generates enormous amounts of data and in order to build smarter, more personalised products, it is required   to analyse these datasets and draw logical conclusions from it. Therefore, performing computations on big data is inevitable,  
and efficient   algorithms that are able to deal with large amounts of data, are the need of the day. 

We would like to emphasize that most of these datasets are high dimensional and sparse -- the number of possible attributes in the dataset are large however only a small number of them are present in most of the data points. For example: 
   micro-blogging site Twitter can have  each tweet  of maximum $140$ characters.  If we consider only English tweets, considering   the vocabulary size is of  $171,476$ words, each tweet can be represented  as a sparse binary vector  in $171,476$ dimension, where $1$ indicates that a word is present, $0$ otherwise.  Also, large variety of short and irregular forms in tweets  add further sparseness.
Sparsity is also quite common in web documents, text, audio, video data as well.

Therefore, it is desirable to  investigate  the compression techniques that can compress the dimension of the data while preserving the similarity between data objects. 
In this work, we focus on sparse, binary data, which can also be considered as sets, and the underlying similarity measure as Jaccard similarity. Given two sets $A$ and $B$ the Jaccard similarity between them is denoted as $\JS(A, B)$ and is defined as $\JS(A, B)={|A\cap B|}/{|A\cup B|}$. 
Jaccard Similarity is popularly used to determine whether two documents are similar.
\citep{setcontainment} showed that this problem can be reduced to set intersection problem via \textit{shingling \footnote{A document is a string of characters. A $k$-shingle for a document is defined as a contiguous substring of length $k$ found within the document. For example: if our document is ${abcd}$, then shingles of size $2$ are $\{ab, bc, cd\}$.}}. For example: two documents $A$ and $B$ first get converted into two shingles $S_A$ and $S_B$, then similarity between these two documents is defined as 
$\JS(A, B)=|S_A \cap S_B|/|S_A \cup S_B|$. Experiments validate that high Jaccard similarity implies that two documents are similar. 
 
 Broder 
\textit{et al.}~\cite{Broder00,BroderCFM98} suggested a 
technique to compress a collection of sets while preserving the  Jaccard similarity between every pair of sets.
For a set $\mathrm{U}$ of binary vectors $\{\mathbf{u_i}\}_{i=1}^n\subseteq \{0, 1\}^d$, their
 technique includes taking a random permutation 
of $\{1, 2, \ldots, d\}$ and assigning a value to each set which maps to  
minimum  under that permutation. 
\begin{definition}[Minhash~\cite{Broder00,BroderCFM98}]\label{defn:minwise}
    Let $\pi$ be a permutations over $\{1, \ldots, d\}$, then for a set $\mathbf{u}\subseteq \{1,\ldots d\}$
    $h_\pi(\mathbf{u}) = \arg\min_i \pi(i)$ for $i \in \mathbf{u}$. Then,
    \begin{align*}\label{eq:cosine}
 \Pr[h_\pi(\mathbf{u})=h_\pi(\mathbf{v})]=\frac{|\mathbf{u}\cap \mathbf{v}|}{|\mathbf{u} \cup \mathbf{v}|}.
\end{align*}
\end{definition}

\begin{note}[Representing sets as binary vectors]
Throughout this paper, for convenience of notation we represent sets as binary vectors. Let the cardinality of the universal set is $d$, then each set which is a subset of the universal set is represented as a binary vector in $d$-dimension. We mark $1$ at position where the corresponding element from universal set is present, and $0$ otherwise. We illustrate this with an example:  let the universal set is $\mathcal{U}=\{1, 2, 3, 4, 5 \}$, then we represent the set $\{1, 2\}$ as $ 11000$, and the  set $\{1,5\}$  as $100001$.
\end{note}

 \subsection{Revisiting Compression Scheme of~\cite{KulkarniP16}}
 Recently, Pratap and Kulkarni~\cite{KulkarniP16} suggested a compression scheme for binary data that compress the data while preserving both Hamming distance and Inner product. A major advantage of their scheme is that the compression-length depends only on the sparsity of the data and is independent of the dimension of data.  In the following we briefly discuss  their compression scheme:

Consider a set of $n$ binary vectors in $d$-dimensional space, then,  
given a binary vector $\textbf{u}\in \{0,1\}^{d}$, their scheme compress it into a 
$\N$-dimensional binary vector (say) $\mathbf{u'}\in\{0,1\}^{\N}$ as follows, where $\N$ to be specified later. 
It randomly assign each bit position (say) $\{i\}_{i=1}^d$ of the original
data to an integer $\{j\}_{j=1}^{\N}$. Further, to compute the $j$-th bit of the compressed vector $\mathbf{u'}$ 
we check which bits positions have been mapped to $j$,  and compute 
the parity of bits located at those positions, and assign it to the $j$-th bit  position. 
Figure~\ref{fig:bcs} illustrate this with an  example, and the definition below state is more formally. 
In continuation of their analogy we call it as $\BCS.$
 \begin{figure}[ht!]
\centering
\includegraphics[scale=.03]{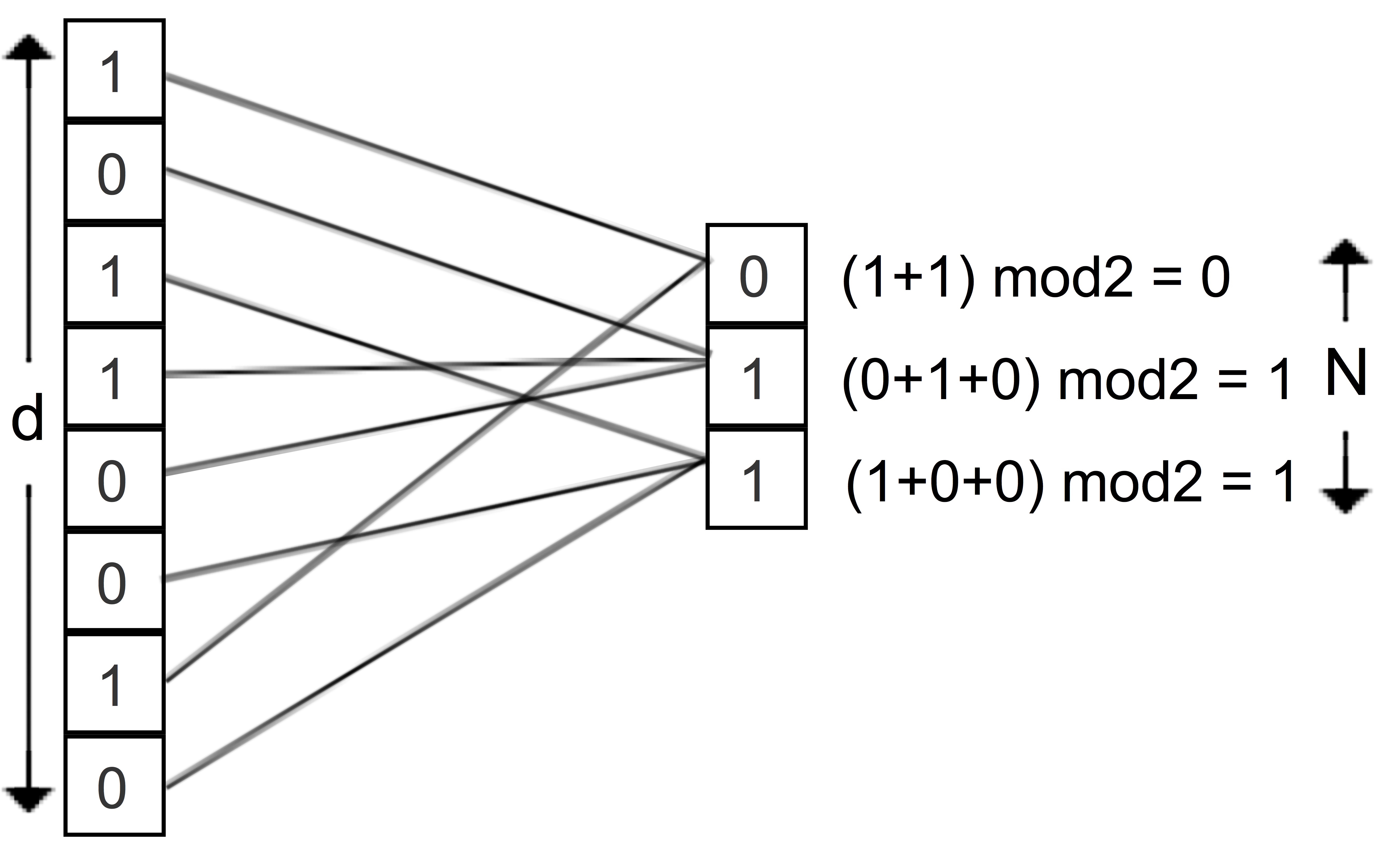}
\caption{Binary Compression Scheme ($\BCS$) of~\citep{KulkarniP16}}
\label{fig:bcs}
\end{figure}

\begin{definition}[Binary Compression Scheme -- $\BCS$ (Definition $3.1$ of~\cite{KulkarniP16}) ]\label{defi:bcs}
 Let $\N$ be the number of buckets (compression length), for $i=1$ to $d$, we randomly assign 
 the $i$-th  position to a bucket number $b(i)$  $\in \{1, \ldots \N\}$. Then a vector $\mathbf{u} \in \{0, 1\}^d$,  
 compressed into a vector $\mathbf{u}' \in \{0, 1\}^{\N}$ as follows:
 \[\mathbf{u}'[j] = \sum_{i : b(i) = j} \mathbf{u}[i]  \pmod 2.\]  
\end{definition}

 \subsection{Our Result}
 Using the above mentioned  compression scheme, we are able to prove the following compression guarantee for Jaccard similarity. 
 
\begin{theorem}\label{theorem:compressionJS}
Consider a set $\mathrm{U}$ of binary vectors $\{\mathbf{u_i}\}_{i=1}^n\subseteq \{0, 1\}^d$ with maximum number of $1$ in any vector is at most $\tb$,
 a positive  integer $r$,    $\epsilon\in (0, 1)$, and $\varepsilon \geq \max \{ \epsilon, 2\epsilon/1-\epsilon  \}$.   
    We set $\N=O({\tb}^2\log^2n) $,  and compress them into 
    a set  $\mathrm{U'}$ of binary vectors
    $\{\mathbf{u_i'}\}_{i=1}^n\subseteq\{0, 1\}^{\N}$  \textit{via} $\BCS$.   
    Then for all $\mathbf{u_i}, \mathbf{u_j}\in \mathrm{U}$ the following holds with probability
 at least $1-{2}/{n}$,
 \[
  (1-\varepsilon)\JS(\mathbf{u_i}, \mathbf{u_j})   \leq  \JS({\mathbf{u_i}}', {\mathbf{u_j}}')  \leq (1+\varepsilon) \JS(\mathbf{u_i}, \mathbf{u_j}).
 \]
 \end{theorem}
 
 \begin{remark}
 A major benefit (as also mentioned in~\cite{KulkarniP16}) of $\BCS$ is that it also works well in the streaming setting. The only prerequisite is an upper bound on the sparsity $\tb$ as well as on the number of data points, which requires to give a bound on the compression length $\N$. 
 \end{remark}

\subsection*{Parameters for evaluating a compression scheme}
The quality of a compression algorithm can be evaluated on the following parameters. 
\begin{itemize}
\item \textit{Randomness} is the  number of random bits  required for   compression.
\item \textit{Compression time} is the   time required for compression.
\item \textit{Compression length} is the  dimension of   data after  compression. 
\item The amount of \textit{space} required to store the compressed matrix. 
\end{itemize}
Ideally  the compression length and the compression time should be as small as possible while keeping the accuracy as high as possible.

 \subsection{Comparison between $\BCS$ and minhash}
 We evaluate the quality of our compression scheme with minhash on the parameters stated earlier. 
 \paragraph{Randomness} One of the major advantages of BCS is the reduction in the number of random bits required for compression. 
 We quantify it  below.
 \begin{lemma}
  Let  a set of $n$ $d$ dimensional binary vectors, which get compressed into a set of $n$ vectors in $\N$ dimension \textit{via} minhash and $\BCS$, respectively. Then, the amount of random bits required for  $\BCS$ and minhash are    $O(d \log \N)$ and $O(\N d \log d)$, respectively.
 \end{lemma}
    
 \begin{proof}
For $\BCS$, it is required to map each bit position from $d$-dimension to $\N$-dimension. One bit assignment requires 
 $O(\log \N)$ amount of randomness as it needs to generate a number between $1$ to $\N$ which require $O(\log \N)$ bits. 
 Thus, for each   bit position in $d$-dimension, the mapping requires $O(d\log \N)$ amount of randomness.  On the other side, minhash   requires creating $\N$ permutations in $d$-dimension. One permutation in $d$ dimension requires generating $d$ random numbers each within $1$ and $d$. Generating a number between $1$ and $d$ requires $O(\log d)$ random bits, and  generating $d$ such numbers require $O(d \log d)$ random bits. Thus,  generating $\N$ such random permutations requires $O(\N d \log d)$ random bits. 
 
\end{proof}
\paragraph{Compression time} 
$\BCS$ is significantly faster than Minhash algorithm in terms of compression time. This is because, generation of random bits requires a considerable amount of time. Thus, reduction in compression time is proportional to the reduction in the amount of randomness required for compression. Also, for compression length $\N$, Minhash scans the vector $\N$ times - once for each permutation, while $\BCS$ just requires a single scan.


\paragraph{Space required for compressed data:} Minhash compression generates an integer matrix as opposed to the binary matrix generated by $\BCS$. Therefore, the space required to store the compressed data of $\BCS$ is $O(\log d)$ times less as compared to minhash.

\paragraph{Search time} Binary form of our compressed data leads to a significantly faster search as efficient bitwise operations can be used.

In Section~\ref{sec:experiments}, we numerically quantify the advantages of our compression on the later three parameters \textit{via} experimentations on synthetic and real-world datasets. 

Li \textit{et. al.}~\cite{Libit} presented "b-bit minhash"  an improvement over  Broder's minhash by reducing the compression size. They   store  only a vector of $b$-bit hash values (of binary representation) of the corresponding integer hash value.  However, this approach introduces some error in the accuracy. If we compare $\BCS$ with $b$-bit minhash, then we have same the advantage as of minhash in savings of randomness and compression time. Our search time is again better as we only store one bit instead of $b$-bits.

\subsection{Applications of our result}\label{sec:application}

In cases of high dimensional, sparse data, $\BCS$ can be used to improve numerous applications where currently minhash   is used.

\paragraph{Faster ranking/ de-duplication of documents} Given a corpus of documents and a set of query documents, ranking documents from the corpus based on similarity with the query documents is an important problem in information-retrieval. This also helps in identifying duplicates, as documents that are ranked high with respect to the query documents, share high similarity. 
Broder~\cite{BroderCPM00} suggested an efficient de-duplication technique for documents -- by converting documents to \textit{shingles };  defining the similarity of two documents based on their Jaccard similarity; and then using minhash sketch to efficiently detect near-duplicates. As most the datasets are  sparse, $\BCS$ can be more effective than minhash on the parameters stated earlier.  
 
 \paragraph{Scalable Clustering of documents:}  Clustering    is   one of the fundamental   information-retrieval problems. \cite{broder2000method} suggested an approach to cluster data objects that are similar. The approach is to partition the data into shingles; defining the similarity of two documents based on their Jaccard similarity; and then via minhash generate a sketch of each data object. These sketches preserve the similarity of data objects.  Thus, grouping these sketches  gives a clustering on the original documents. However, when documents are high dimensional such as webpages, minhash sketching approach  might not be efficient. 
 Here also, exploiting the sparsity of documents $\BCS$  can be more effective  than minhash. 
 
 Beyond above applications, minhash compression has been widely used in applications like spam detection~\cite{broder1997resemblance}, compressing social networks~\cite{ChierichettiKLMPR09}, all pair similarity~\cite{BayardoMS07}. 
 As in most of these cases, data objects are sparse, $\BCS$  can provide almost accurate and more efficient solutions to these problems.  
 
 We experimentally validate the performance of $\BCS$ for raking experiments on UCI~\citep{UCI} "BoW" dataset and achieved significant improvements  over minhash. We discuss this in Subsection~\ref{subsection:realworld}. Similarly, other mentioned results can also be validated. 


\paragraph {Organization of the paper}
Below, we first present some necessary notations that 
are used in the paper. 
In Section~\ref{sec:analysis}, we first  revisit the results of~\cite{KulkarniP16}, then building on it we give a proof on the compression bound for Jaccard similarity.  In Section~\ref{sec:experiments}, we complement our theoretical results \textit{via} extensive experimentation on synthetic as well as real-world datasets. 
Finally, in Section~\ref{sec:conclusion} we conclude our discussion and state some open questions.
 \begin{tabular}{|c|l|}
\hline
 \multicolumn{2}{|c|}{\bf Notations}\\
 \hline
 $\N$ & dimension of  the compressed data \\
 \hline
$\tb$ & upper bound on the number of $1$'s in binary data.\\
\hline
$\mathbf{u}[i]$ & $i$-th bit position  of   vector $\mathbf{u}.$\\
\hline
$\JS(\mathbf{u}, \mathbf{v})$ & Jaccard similarity between   binary vectors $\mathbf{u}$ and $\mathbf{v}.$\\
\hline
 $\dH(\mathbf{u}, \mathbf{v})$& Hamming distance between binary  vectors $\mathbf{u}$ and $\mathbf{v}.$\\
 \hline
$\langle\mathbf{u}, \mathbf{v}\rangle$ & Inner product between binary vectors $\mathbf{u}$ and $\mathbf{v}.$\\
\hline
 \end{tabular}

\section{Analysis}\label{sec:analysis}
We first revisit  the  results of~\cite{KulkarniP16} which discuss compression bounds for  Hamming distance and Inner product, and then building on it, we give a compression   bound for  Jaccard similarity. We  start with discussing the intuition and a proof sketch of their result.

Consider  two binary vectors $\mathbf{u}, \mathbf{v} \in\{0, 1\}^d$, we call a bit position 
 \textit{``active''} if at least one of the vector between $\mathbf{u}$ and $\mathbf{v}$ has value $1$ in that  position. Further, given the sparsity bound $\tb$, there can be at most  
 $2\tb$ active positions   between  $\mathbf{u}$ and $\mathbf{v}$. Then let via $\BCS$, they  
compressed into binary vectors $\mathbf{u'}, \mathbf{v'} \in \{0, 1\}^{\N}$.
In the compressed version, we call a bit position \textit{``pure''} if the number of active positions mapped to it is at most one, and   \textit{"corrupted"} otherwise. 
The contribution of pure bit 
positions in $\mathbf{u'}, \mathbf{v'}$ towards Hamming distance (or Inner product similarity), is exactly equal to the 
contribution of the bit positions in $\mathbf{u}, \mathbf{v}$ which get mapped to the pure bit positions. 
 \begin{figure}[h!]
\centering
\includegraphics[scale=.025]{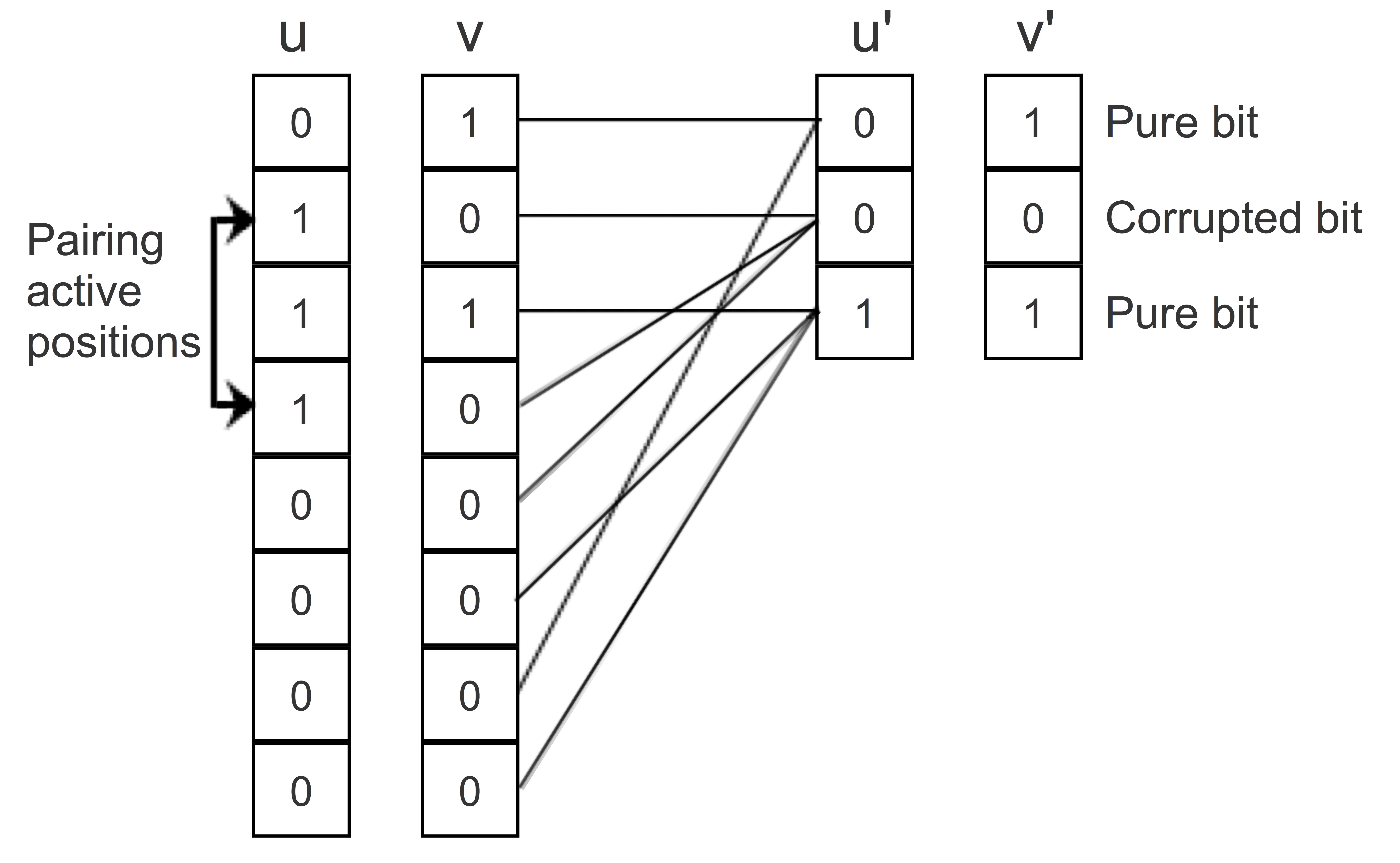}
\caption{Illustration of pure/corrupted bits in $\BCS$.}
\label{fig:pure}
\end{figure}
Further, the  deviation of Hamming distance (or Inner product similarity) between $\mathbf{u'}$ and $\mathbf{v'}$ from that 
of $\mathbf{u}$ and $\mathbf{v}$, corresponds to the number of corrupted bit positions shared between 
$\mathbf{u'}$ and $\mathbf{v'}$. Figure~\ref{fig:pure} illustrate this with an example, and the lemma below analyse it.
 \begin{lem}[Lemma $14$ of ~\cite{KulkarniP16}]\label{lem:compression}
 Consider two binary vectors $\mathbf{u}, \mathbf{v} \in\{0, 1\}^d$, which get compressed into  
 vectors $\mathbf{u'}, \mathbf{v'} \in \{0, 1\}^{\N}$ 
 using the $\BCS$, and  suppose $\tb$ is the maximum number of $1$ in any vector. 
 Then for an  integer $r\geq1$, and 
  $\epsilon \in (0, 1)$, probability that $\mathbf{u'}$ and $\mathbf{v'}$ share more than $\epsilon r$ 
  corrupted positions is at most
 $\left({2\tb}/{\sqrt{\N}}\right)^{\epsilon r}.$
\end{lem}
\begin{proof}
We first give a bound on the probability that a particular bit position gets corrupted between $\mathbf{u'}$ and $\mathbf{v'}$. 
  As there are at most $2\tb$ active positions shared between vectors $\mathbf{u}$ and $\mathbf{v}$, the number of ways of
 pairing two active positions from $2\tb$ active positions is  at most $2\tb \choose 2$, and this 
 pairing will result in a corrupted bit position in $\mathbf{u'}$ or $\mathbf{v'}$. 
 Then, the  probability that a particular bit position in 
 $\mathbf{u'}$ or $\mathbf{v'}$ gets corrupted is at most 
 ${{2\tb \choose 2}}/{{\N}}\leq \left({4{\tb}^2}/{\N} \right).$
 Further, if the deviation of Hamming distance (or Inner product similarity) between $\mathbf{u'}$ and $\mathbf{v'}$ 
 from that of  $\mathbf{u}$ and $\mathbf{v}$ is more than $\epsilon r$, then at least $\epsilon r$ corrupted 
 positions are shared between $\mathbf{u'}$ and $\mathbf{v'}$, 
 which implies that at least $\frac{\epsilon r}{2}$ pair of active positions in $\mathbf{u}$ and $\mathbf{v}$
 got paired up while compression. 
 The number of possible ways of pairing $\frac{\epsilon r}{2}$ active positions from $2\tb$ active positions 
 is at most ${2\tb \choose \frac{\epsilon r}{2}}{2\tb-\frac{\epsilon r}{2} \choose \frac{\epsilon r}{2}} \frac{\epsilon r}{2}!\leq
 (2\tb)^{\epsilon r}.$ 
 Since the probability that a pair of active positions got mapped in the same bit position in the compressed data 
 is $\frac{1}{\N}$, the probability that  $\frac{\epsilon r}{2}$  pair of active positions got mapped 
 in $\frac{\epsilon r}{2}$ distinct bit positions in the compressed data is at 
 most $(\frac{1}{\N})^{\frac{\epsilon r}{2}}$. 
   Thus, by union bound, the probability that at least $\epsilon r$ corrupted bit 
 position shared between $\mathbf{u'}$ and $\mathbf{v'}$  is at most ${(2\tb)^{\epsilon r}}/({{\N}^{{\epsilon r}/{2}}})
 =\left({2\tb}/{\sqrt{\N}}\right)^{\epsilon r}.$
 \end{proof}
 \vspace{-0.4cm}
 In the lemma below   generalise the above result for a set of $n$ binary vectors, and suggest a compression bound so that 
 any pair of compressed vectors  share only a  very small number of corrupted bits, with high probability.  
 \begin{lem}[Lemma $15$ of ~\cite{KulkarniP16}]\label{lem:compressionBound}
  Consider a set~$\mathrm{U}$ of $n$  binary vectors $\{\mathbf{u_i}\}_{i=1}^n\subseteq \{0, 1\}^d$, 
  which get compressed into  a set  $\mathrm{U'}$ of binary
  vectors $\{\mathbf{u_i'}\}_{i=1}^n\subseteq\{0, 1\}^{\N}$  using the $\BCS$. Then for any 
  positive  integer $r$, and   $\epsilon \in (0, 1)$, 
 \begin{itemize}
  \item if $\epsilon r >3 \log n$, and  we set $\N=O({\tb}^2)$, then probability 
  that   for all $\mathbf{u_i'}, \mathbf{u_j'}\in \mathrm{U'}$
  share more than $\epsilon r$   corrupted positions is at most  ${1}/{n}$. 
  \item If $\epsilon r < 3 \log n$, and  we set $\N=O({\tb}^2\log^2n) $, then probability 
  that   for all $\mathbf{u_i'}, \mathbf{u_j'}\in \mathrm{U'}$
  share more than $\epsilon r$   corrupted positions is at most  ${1}/{n}$.
 \end{itemize}
 \end{lem}
   \begin{proof}
 For a fixed pair of compressed 
 vectors $\mathbf{u_i'}$ and $\mathbf{u_j'}$, due to lemma~\ref{lem:compression}, probability that they
 share more than $\epsilon r$ corrupted positions is at most  $\left({2{\tb}}/{\sqrt{\N}}\right)^{\epsilon r}.$
 If $\epsilon r >3 \log n$, and  $\N=16{\tb}^2$, then the above probability is at most 
 $\left({2\tb}/{\sqrt{\N}}\right)^{\epsilon r}<\left({2{\tb}}/{4t}\right)^{3 \log n}=\left({1}/{2}\right)^{3 \log n}<{1}/{n^3}.$ 
 As there are at most 
 ${n \choose 2}$ pairs of vectors,  the required bound follows from union bound of probability. 
  
  In the second case,   as $\epsilon r <3 \log n$, we cannot  bound the  probability as above. 
 Thus, we replicate each bit position $3 \log n$ times, 
  which makes a $d$ dimensional 
 vector to a $3d\log n$  dimensional, and as a consequence the  Hamming distance 
 (or Inner product similarity) is also scaled up by a multiplicative factor of $3 \log n$.
 We now apply the compression scheme on these scaled vectors, then for a fixed pair of compressed 
 vectors $\mathbf{u_i'}$ and $\mathbf{u_j'}$,  probability that they  have more than 
 $3 \epsilon r \log n $ corrupted positions is at most    
 $\left({6\tb\log n }/{\sqrt{\N}}\right)^{3 \epsilon r\log n}$. As we set 
 $\N=144{\tb}^2\log^2n $, the above probability is at most 
 $\left({6\tb\log n }/{\sqrt{144{\tb}^2\log^2n}}\right)^{3 \epsilon r \log n}< \left({1}/{2}\right)^{3 \log n}<{1}/{n^3}.$
 The final probability follows due to union bound over all ${n \choose 2}$ pairs.
  \end{proof}
  After compressing binary data \textit{via} $\BCS$, the Hamming distance between any pair of binary vectors can not increase. This is due to the fact that  compression doesn't generate any new $1$ bit, which could increase the Hamming distance from the   uncompressed version. 
  In the following, we recall  the main result of~\citep{KulkarniP16}, which holds due the above fact and Lemma~\ref{lem:compressionBound}.
  \begin{theorem}[Theorem $1$, $2$ of~\cite{KulkarniP16}]\label{theorem:compressionHammingIP}
Consider a set $\mathrm{U}$ of binary vectors $\{\mathbf{u_i}\}_{i=1}^n\subseteq \{0, 1\}^d$, 
  a positive  integer $r$, and   $\epsilon \in (0, 1)$. 
  If $\epsilon r >3 \log n$,  we set $\N=O({\tb}^2)$;  if $\epsilon r < 3 \log n$, 
    we set $\N=O({\tb}^2\log^2n) $,  and compress them 
  into  a set  $\mathrm{U'}$ of binary vectors $\{\mathbf{u_i'}\}_{i=1}^n\subseteq\{0, 1\}^{\N}$  \textit{via} $\BCS$.
    Then for all  $\mathbf{u_i}, \mathbf{u_j}\in \mathrm{U}$, 
\begin{itemize}
 \item if $\dH(\mathbf{u_i}, \mathbf{u_j})< r$, then $\Pr [\dH({\mathbf{u_i}}', {\mathbf{u_j}}')< r]=1$,
  \item if $\dH(\mathbf{u_i}, \mathbf{u_j})\geq (1+\epsilon)r$, then $\Pr [\dH({\mathbf{u_i}}', {\mathbf{u_j}}')< r]<{1}/{n}.$
\end{itemize}
 For Inner product, if we set $\N=O({\tb}^2\log^2n) $, then the following is  true with probability
 at least $1-{1}/{n}$,
 \begin{itemize}
 \item $(1-\epsilon) \langle\mathbf{u_i}, \mathbf{u_j} \rangle  \leq \langle {\mathbf{u_i}}', {\mathbf{u_j}}'\rangle \leq (1+\epsilon)\langle\mathbf{u_i}, \mathbf{u_j}\rangle.$
 \end{itemize}
\end{theorem}

  The following proposition relates Jaccard similarity with Inner product and Hamming distance. The proof follows as for a pair binary vectors their Jaccard similarity is the ratio of the number of positions where $1$ is appearing together, with the number of bit positions where $1$ is present in either of them. Clearly, numerator is captured by the Inner product  between those pair of vectors, and denominator is captured by Inner product plus Hamming distance between them -- number of positions where $1$ is occurring in both vectors, plus the number of positions where $1$ is present in exactly one of them.
  \begin{prop}\label{prop:JS}
  For any pair of vectors $\mathbf{u}, \mathbf{v} \subseteq \{0, 1\}^d$, we have 
  $
  \JS(\mathbf{u}, \mathbf{v})=\langle  \mathbf{u}, \mathbf{v} \rangle / \left(  \langle   \mathbf{u}, \mathbf{v}  \rangle   + \dH(\mathbf{u}, \mathbf{v}) \right)
  $
  \end{prop}
  
   In the following, we complete a proof of the Theorem~\ref{theorem:compressionJS} due to Proposition~\ref{prop:JS}, and Theorem~\ref{theorem:compressionHammingIP}.
  \paragraph{\textbf{Proof of Theorem~\ref{theorem:compressionJS}}}
   Consider a pair of vectors   $\mathbf{u_i}, \mathbf{u_j}$ from the set $\mathrm{U}\in \{0, 1\}^d$, which get compressed into binary
  vectors $\mathbf{u_i'}, \mathbf{u_j'}\in \{0, 1\}^{\N}$. Due to Proposition~\ref{prop:JS}, we have 
 $
   \JS(\mathbf{u_i'}, \mathbf{u_j'})={\langle\mathbf{u_i'}, \mathbf{u_j'}   \rangle}/({\langle \mathbf{u_i'}, \mathbf{u_j'}   \rangle+\dH(\mathbf{u_i'}, \mathbf{u_j'})}).
$
Below, we present a lower and  upper bound on the expression. 
   \begin{align*}
   \JS(\mathbf{u_i'}, \mathbf{u_j'}) 
   &\geq \frac{(1-\epsilon)\langle\mathbf{u_i}, \mathbf{u_j}   \rangle}{(1-\epsilon)\langle \mathbf{u_i}, \mathbf{u_j}   \rangle+\dH(\mathbf{u_i}, \mathbf{u_j})}\numberthis\label{eq:thmresults1}\\
   &  \geq \frac{(1-\epsilon)\langle\mathbf{u_i}, \mathbf{u_j}   \rangle}{\langle \mathbf{u_i}, \mathbf{u_j}   \rangle+\dH(\mathbf{u_i}, \mathbf{u_j})}\\
   & \geq (1-\varepsilon) \JS(\mathbf{u_i}, \mathbf{u_j})\numberthis\label{eq:lowerbound}
  \end{align*}
  Equation~\ref{eq:thmresults1} holds hold true with probability at least $1-1/n $ due to Theorem~\ref{theorem:compressionHammingIP}.
  \begin{align*}
   \JS(\mathbf{u_i'}, \mathbf{u_j'})&=\frac{\langle\mathbf{u_i'}, \mathbf{u_j'}   \rangle}{\langle \mathbf{u_i'}, \mathbf{u_j'}   \rangle+\dH(\mathbf{u_i'}, \mathbf{u_j'})}\\
   &\leq \frac{(1+\epsilon)\langle\mathbf{u_i}, \mathbf{u_j}   \rangle}{(1+\epsilon)\langle \mathbf{u_i}, \mathbf{u_j}   \rangle+(1-\epsilon)\dH(\mathbf{u_i}, \mathbf{u_j})} \numberthis\label{eq:thmresults}\\
   &=\left( \frac{1+\epsilon}{1-\epsilon}\right)\frac{\langle\mathbf{u_i}, \mathbf{u_j}   \rangle}{\left(\frac{(1+\epsilon)}{(1-\epsilon)}\langle \mathbf{u_i}, \mathbf{u_j}   \rangle+\dH(\mathbf{u_i}, \mathbf{u_j})\right)}\\
    &\leq\left( \frac{1+\epsilon}{1-\epsilon}\right)\frac{\langle\mathbf{u_i}, \mathbf{u_j}   \rangle}{\langle \mathbf{u_i}, \mathbf{u_j}   \rangle+\dH(\mathbf{u_i}, \mathbf{u_j})}\\
      &=\left( 1+\frac{2\epsilon}{1-\epsilon}\right)\frac{\langle\mathbf{u_i}, \mathbf{u_j}   \rangle}{\langle \mathbf{u_i}, \mathbf{u_j}   \rangle+\dH(\mathbf{u_i}, \mathbf{u_j})}\\
        &\leq\left( 1+\varepsilon\right)\frac{\langle\mathbf{u_i}, \mathbf{u_j}   \rangle}{\langle \mathbf{u_i}, \mathbf{u_j}   \rangle+\dH(\mathbf{u_i}, \mathbf{u_j})} \numberthis\label{eq:epsbound}\\
        &=(1+\varepsilon) \JS(\mathbf{u_i}, \mathbf{u_j}) \numberthis\label{eq:upperbound}
  \end{align*}
  Equation~\ref{eq:thmresults} holds hold true with probability at least $(1-1/n)^2\geq 1-2/n $ due to Theorem~\ref{theorem:compressionHammingIP}; Equation~\ref{eq:epsbound} holds as $\varepsilon\geq \frac{2\epsilon}{1-\epsilon}.$
    Finally Equations~\ref{eq:upperbound} and \ref{eq:lowerbound} complete a proof of the Theorem.
\section{Experimental Evaluation}\label{sec:experiments}
We performed our experiments on a machine having the following configuration: 
CPU: Intel(R) Core(TM) i5 CPU @ 3.2GHz x 4; 
Memory: 8GB 1867 MHz DDR3; 
OS: macOS Sierra 10.12.5;  
OS type: 64-bit. 
We performed our experiments on synthetic and real-world datasets, we discuss them one-by-one as follows:
\subsection{Results on Synthetic Data} 
We performed two experiments on  synthetic dataset and showed that it preserves both: a) all-pair-similarity,  and b) $k$--$\NN$ similarity. 
In all-pair-similarity, given a set of $n$ binary vectors in $d$-dimensional space with the sparsity bound 
   $\tb$, we showed that after compression Jaccard similarity    between every pair of vector is   preserved. In $k$--$\NN$ similarity, given is a query vector $\Sq$, we showed that after compression Jaccard similarity  between $\Sq$ and  the  vectors  that are similar to $\Sq$,  are   preserved. We performed   experiments on  dataset consisted of $1000$ vectors in $100000$ dimension. Throughout synthetic data experiments, we calculate the accuracy \textit{via} Jaccard ratio, that is, if the set $\mathcal{O}$ denotes  the ground truth result,  and the set $\mathcal{O'}$ denotes our result,  then the accuracy of our result is calculated by  the Jaccard ratio between the sets $\mathcal{O}$ and $\mathcal{O'}$ -- that is $\JS(\mathcal{O}, \mathcal{O'})={|\mathcal{O}\cap \mathcal{O'}|}/{|\mathcal{O}\cup \mathcal{O'}|}$. To reduce the effect of randomness we repeat the experiment $10$ times and took  average. 
\subsubsection{Dataset generation}
\paragraph{All-pair-similarity} We generated
 $1000$ binary vectors in dimension  $100000$ such that the sparsity of each vector is at most $\tb$. 
  If we randomly choose binary vectors  respecting the sparsity bound, then most likely every pair of vector will have similarity $0$. Thus, we had to deliberately generate some vectors having   high  similarity. We  generated $200$ pairs whose similarity is high.  
  To generate such a pair, we  choose a random number (say $s$) between $1$ and $\tb$, then we randomly select those many position (in dimension) from $1$ to $100000$,  set $1$ in both of them, and set  remaining to $0$. Further, for each of the vector in the pair, we choose a random number (say $s'$) from the range $1$ to $ (\tb - s)$, and again randomly sample those many positions from the remaining positions and set them to $1$.  
  This gives a pair of vectors having  similarity  at least $\frac{s}{s+2s'}$  and respecting the sparsity bound. We repeat this step  $200$ times and obtain   $400$ vectors.  For each of the  remaining $600$ vectors, we randomly 
 choose an integer from the range $1$ to $\tb$,    choose those many positions in the dimension,   set them to $1$, and set the remaining positions to $0$.  Thus, we obtained $1000$ vectors of dimension $100000$ which we used as an
input matrix. 

\paragraph{$k$--$\NN$-- similarity}
  A dataset for this experiment consist of a random query vector $\Sq$; $249$ vectors whose similarity with $\Sq$ is high;  and $750$ other vectors. 
We first generated a query vector $\Sq$ of sparsity between $1$ and $\tb$, then using the procedure mentioned above we generated $249$ vectors whose similarity with $\Sq$ is high. Then we generated $750$ random 
  vectors of sparsity is at most  $\tb$. 
  
  \paragraph{Data representation}  We can imagine synthetic dataset as a binary matrix of dimension $100000\times 1000$. However, for ease and efficiency of implementation, we use a compact representation which consist of  a list of lists. The the number of lists is equal to the number of vectors in the binary matrix, and within  each list we just store the indices (co-ordinate) where $1$s are present.  We use this list as an input for both $\BCS$ and minhash.

  \subsubsection{Evaluation metric} We performed two experiments on synthetic dataset -- 1) fixed  sparsity   while  varying compression length,  
  and 2)   fixed compression length while  varying sparsity.
  We present these experimental results  in 
   Figures~\ref{fig:CL}, \ref{fig:Sparsity} respectively. In both of these experiments, we compare and contrast  the performance $\BCS$ with minhash on \textit{accuracy, compression time,} and \textit{search time} parameters. 
All-pair-similarity experiment result requires a quadratic search -- generation of all possible candidate pairs and then pruning those whose similarity score is high, and  the corresponding search time is the time required to compute  all such pairs. While $k$--$\NN$-- similarity experiment requires a linear search and pruning with respect to the query vector $\Sq$, and the corresponding  search time is the time required to compute  such vectors.
   In order to calculate the accuracy on a given support threshold value, we first run a simple brute-force  search algorithm on the entire (uncompressed) dataset, and obtain the ground truth result. 
   Then, we calculate the  Jaccard ratio between  our algorithm's result/  minhash's result,  with the corresponding exact result, and compute the accuracy. 
First row of the plots are     "accuracy" \textit{vs} "compression length/sparsity".
   The second row of the plots are    "compression time" \textit{vs} "compression length/sparsity". 
   Third row of plot shows comparison with respect to "search time" \textit{vs} "compression length/sparsity". 
   
\begin{figure*}[h]
\centering
\includegraphics[scale=.51]{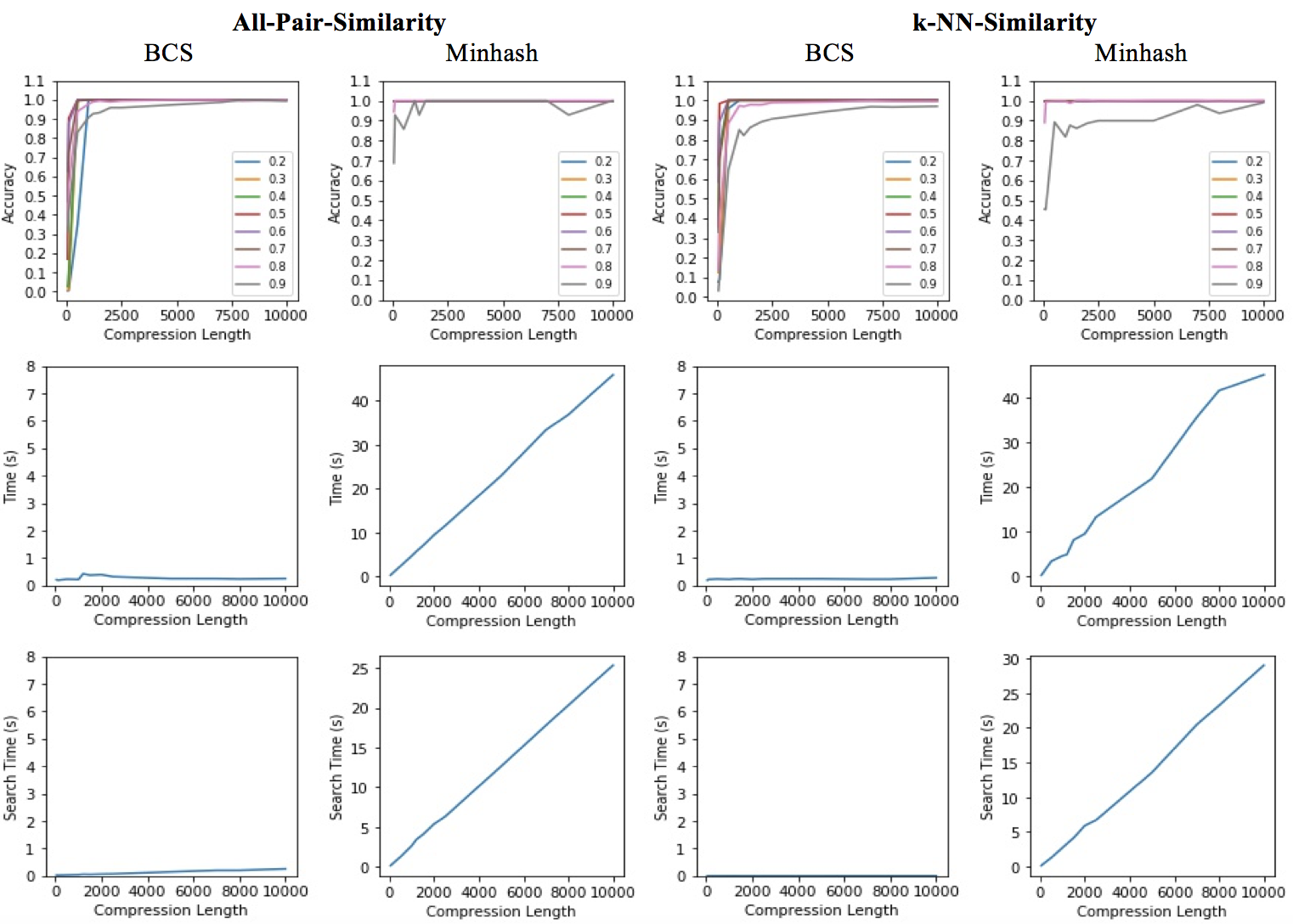}
\caption{Experiments on Synthetic Data:  for  fixed  sparsity $\tb=200$ and varying compression length}
\label{fig:CL}
\end{figure*}

\begin{figure}[h]
\centering
\includegraphics[scale=.51]{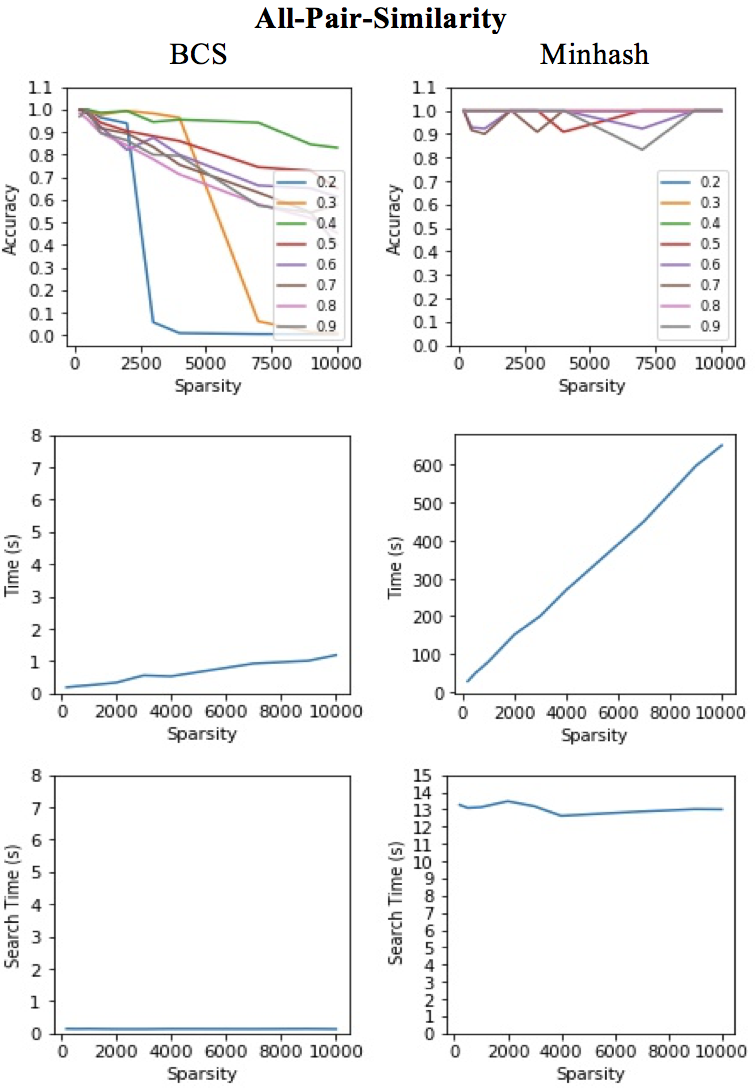}
\caption{Experiments on Synthetic Data:  for  fixed  compression length $5000$ and varying sparsity}
\label{fig:Sparsity}
\end{figure}

   \subsubsection{Insight}

In Figure~\ref{fig:CL}, we plot the result of $\BCS$ and minhash for all-pair-similarity and $k$--$\NN$-- similarity. For this experiment, we fix the sparsity $\tb=200$ and generate the datasets as stated above. We compress the datasets using $\BCS$ and minhash for a range of compression lengths from $50$ to $10000$. It can be observed that $\BCS$ performs remarkably well on the parameters of compression time and search time. Our compression time remains almost constant at $0.2$ seconds in contrast to the compression time of minhash, which grows linearly   to almost $50$ seconds. On an average, $\BCS$ is $90$ times faster than minhash. Also accuracy for $\BCS$ and minhash is almost equal above compression length $300$, but in the window of $50-300$ minhash performs slightly better than $\BCS$. Further, the search-time on $\BCS$ is also significantly less than minhash for all compression lengths. On an average search-time is $75$ times less than the corresponding minhash search-time. We obtain similar results for $k$--$\NN$-- similarity experiments.


 In  Figure~\ref{fig:Sparsity}, we plot the result of $\BCS$ and  minhash  for all-pair-similarity.  For this experiment, we generate datasets for different values of sparsity ranging from $50$ to $10000$. We compress these datasets using $\BCS$ and minhash to a fixed value of compression length $5000$. In all-pair-similarity, when sparsity value is below $2200$, average accuracy of BCS is above $0.85$. It starts decreasing after that value, at sparsity value is $7500$, the accuracy of $\BCS$ stays above $0.7$, on most of the threshold values. The compression time of $\BCS$ is always below $2$ seconds while compression time of minhash grows linearly with sparsity -- on an average compression time of $\BCS$ is around $550$ times faster than the corresponding minhash compression time. Further, we again significantly reduce search time  -- on an average our search-time is $91$ times less than minhash. We obtain similar results for $k$--$\NN$-- similarity experiments.

\subsection{Results on Real-world Data}\label{subsection:realworld}
\subsubsection{Dataset Description:} We compare the performance of $\BCS$ with minhash on the task of retrieving top-ranked elements based on Jacquard similarity.  We performed this experiment on publicly available high dimensional sparse dataset of UCI machine learning repository~\cite{UCI}.
We used four publicly available dataset from UCI repository - namely, NIPS full papers, KOS blog entries, Enron Emails, and  NYTimes news articles. These datasets are binary "BoW" representation of the corresponding text corpus. We consider each of these datasets as a binary matrix, where each document corresponds to a binary vector,  that is  if a particular word is present in the document, then the corresponding entry is $1$ in that position, and it is $0$ otherwise. For our experiments, we consider the entire corpus of NIPS and KOS dataset, while for ENRON and NYTimes we take a uniform sample of $10000$ documents from their corpus. 
We mention their cardinality, dimension, and sparsity in Table~\ref{tab:UCI}.
\begin{table}
  \caption{Real-world dataset description}
  \label{tab:UCI}
  \begin{tabular}{cccl}
    \toprule
    Data Set &No. of   points&Dimension &Sparsity\\
    \midrule
    NYTimes news articles & $10000$ & $102660$ & $ 871$\\
    Enron Emails & $10000$ & $28102$ & $2021 $\\
    NIPS full papers: & $1500$ & $12419$ & $914$\\
     KOS blog entries & $3430$ & $6906$ & $457$\\
  \bottomrule
\end{tabular}
\end{table}

\subsubsection{Evaluation metric:} 
We split the dataset in two parts $90\%$ and $10\%$  -- the bigger partition is use to compress the data, and is referred as the \textit{training partition}, while the second one is use to evaluate the quality of compression and is referred as \textit{querying partition}. We call each vector of the querying partition as query vector.  For each query vector, we compute  the vectors in the training partition whose Jaccard similarity is higher than a certain threshold (ranging from $0.1$ to $0.9$). We first do this on the uncompressed data inorder to find the underlying ground truth result -- for every query vector compute all vectors that are similar to them. Then we compress the entire data,  on various values of compression lengths,  using our compression scheme/minhash. For each  query vector, we calculate the accuracy of  $\BCS$/minhash by taking Jaccard ratio between the set outputted by $\BCS$/minhash, on various values of  compression length, with set outputted a simple linear search algorithm on entire data. This gives us the accuracy of compression of that particular query vector. We repeat this for every vector in the querying partition, and take the average, and we plot the average accuracy for each value in support threshold and compression length.  
We also note down the corresponding compression time on each of the compression length for both $\BCS$ and minhash.   Search time is time required to do a linear search on the compressed data,  we compute the search time for each of the query vector and take the average in the case of both $\BCS$ and  minhash. Similar to synthetic dataset result, we plot the comparison between our algorithm with minhash on following three points -- 1) accuracy \textit{vs} compression length, 2) compression time \textit{vs} compression length, and 3) search time \textit{vs} compression length. 

\begin{figure*}[h]\label{figure:(ENRON+NYTimes)}
\centering
\includegraphics[scale=.51]{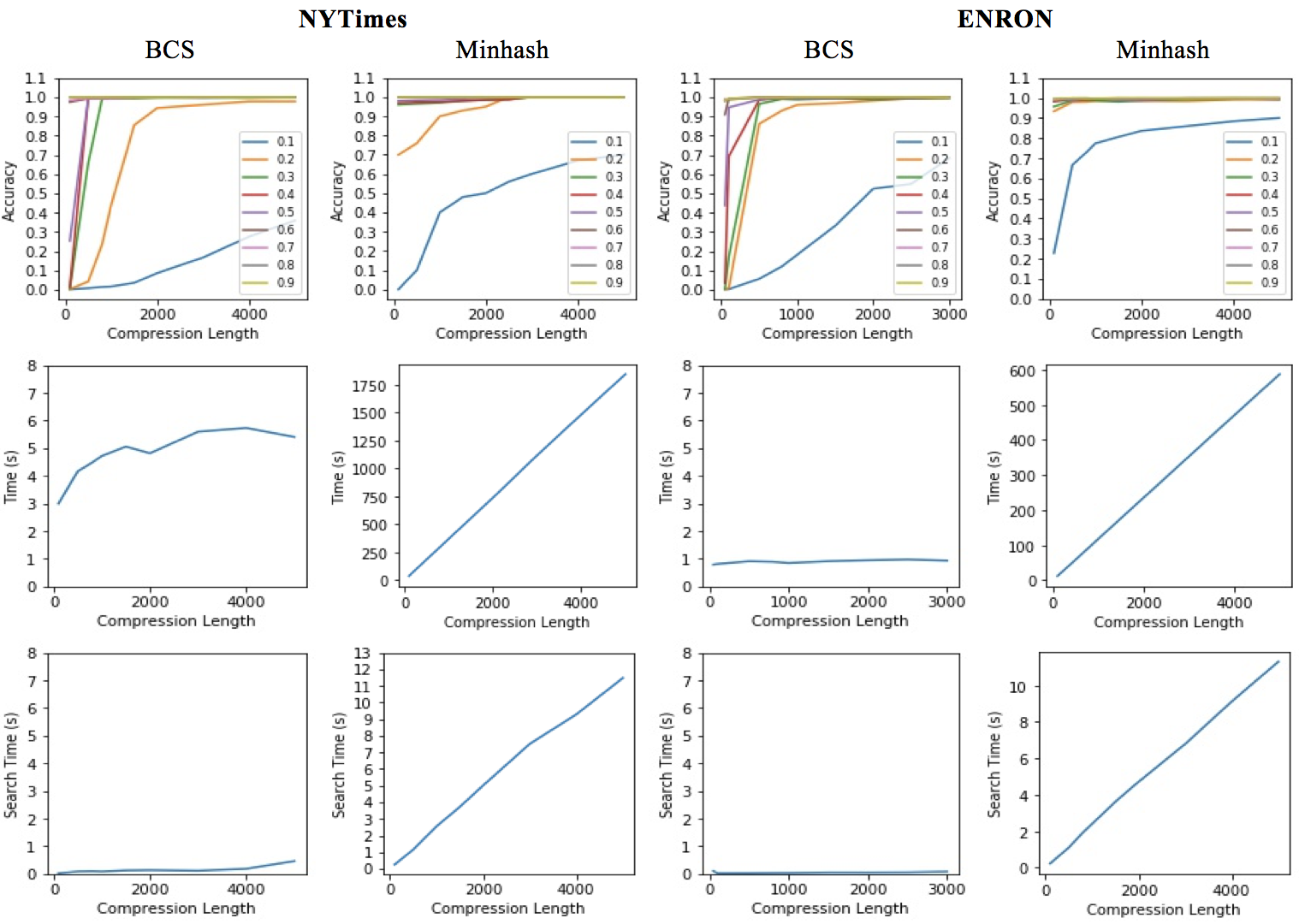}
\includegraphics[scale=.51]{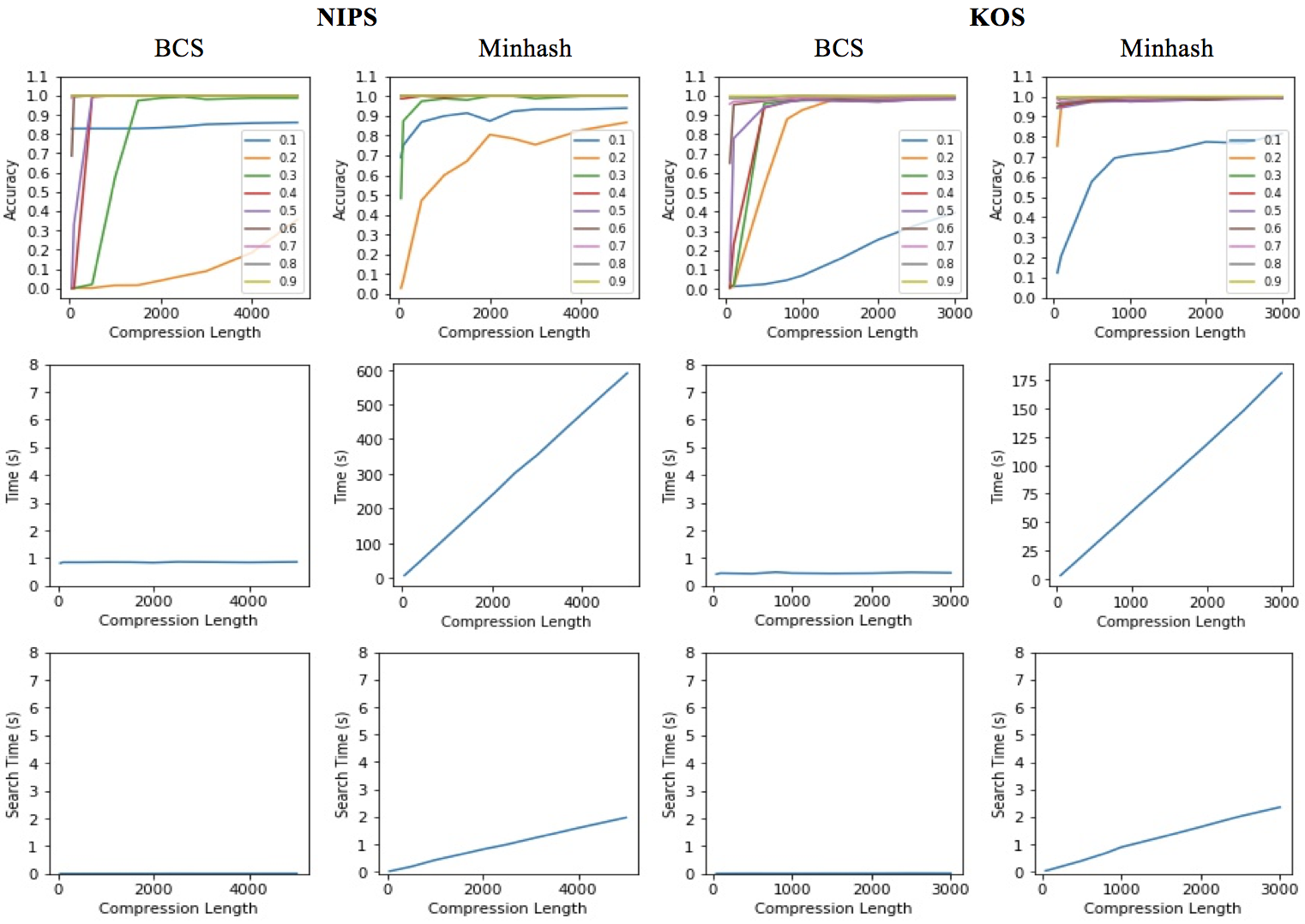}
\caption{Experiments on Real-world datasets~\citep{UCI}. }
\label{fig:realworld}
\end{figure*}

\subsubsection{Insights}

We plot experiments of real world dataset~\cite{UCI} in Figure~\ref{fig:realworld}, and found that performance of $\BCS$ is similar to its performance on synthetic datasets. NYTimes is the sparsest among all other dataset, so the  performance of  $\BCS$ is relatively better as compare to other datasets. For NYTIMES dataset, on an average $\BCS$ is $135$ times faster than minhash, and search time for $\BCS$ is $25$ times less than search time for minhash. For $\BCS$ accuracy starts dropping below $0.9$ when data is compressed below compression length $300$. For minhash, accuracy starts dropping below compression compression length $150$. Similar pattern is observed for ENRON dataset as well, where  $\BCS$ is $268$ times faster than minhash, and a search on the compressed data obtained from $\BCS$ is $104$ times faster than search on data obtained from minhash. KOS and NIPS are dense, low dimensional datasets. However here also, for NIPS, our compression time is $271$ times faster and search-time is $90$ times faster as compared to minhash. For KOS, our compression time is $162$ times faster and search time is $63$ times faster than minhash.

To summarise, $\BCS$ is significantly faster than minhash in terms of both - compression time and search time while giving almost equal accuracy. Also, the amount of randomness required for $\BCS$ is also significantly less as compared to minhash. However, as sparsity is increased, accuracy of $\BCS$ starts decreasing slightly as compared to minhash.

\vspace{-0.3cm}
\section{Concluding remarks  and open questions}\label{sec:conclusion}
We showed that $\BCS$ is able to compress sparse, high-dimensional binary data while preserving the Jaccard similarity. It is considerably faster than the "state-of-the-art" minhash permutation, and  also maintains almost equal accuracy  while significantly reducing the amount of randomness required. Moreover, the compressed representation obtained from $\BCS$ is in binary form, as opposed to integer in case of minhash, due to which the space required to store the compressed data is reduced, and consequently leads to a faster  searches on the compressed representation.  Another major advantage of $\BCS$ is that its compression bound is independent of the dimensions of the data, and only grows polynomially with the sparsity and poly-logarithmically with the number of data points.  We  present a theoretical proof of the same  and complement it with rigorous and extensive experimentations. Our work leaves the possibility of several open questions -- improving the compression bound of our result, and extending it to other similarity measures.

\bibliographystyle{plain}
\bibliography{sample-bibliography} 

\end{document}